\documentclass[10pt, conference, a4paper]{IEEEtran}
\usepackage{amsmath,amsfonts,amssymb,amsbsy, amsthm,ae,aecompl}
\usepackage{algorithm,algpseudocode}
\usepackage[utf8]{inputenc}
\usepackage[english]{babel}
\usepackage{bm}
\usepackage{color}
\usepackage[noadjust]{cite}
\usepackage{epsfig}
\usepackage{enumerate}
\usepackage{float}
\usepackage{fancyhdr}
\usepackage[T1]{fontenc}
\usepackage[acronym,toc,shortcuts]{glossaries}
\usepackage{graphicx, caption, subcaption}
\usepackage{lipsum}
\usepackage{dsfont}
\usepackage{transparent}
\usepackage{tikz,pgfplots}
\usepackage{times}
\usepackage{verbatim}
\usepackage[all]{xy}

\usepackage{etoolbox}
\AtBeginEnvironment{align}{\setcounter{subeqn}{0}}
\newcounter{subeqn} %

\usetikzlibrary{calc}
\usetikzlibrary{arrows,decorations.markings}
\pgfplotsset{
  grid style = {
    dash pattern = on 0.025mm off 0.95mm on 0.025mm off 0mm, 
    line cap = round,
    black,
    line width = 0.5pt
  },
  tick label style={font=\small},
  label style={font=\small},
  legend style={font=\footnotesize},
}

\pgfplotscreateplotcyclelist{laneas-energysto}{
cyan!60!black,		densely dotted, 	every mark/.append style={fill=cyan!80!black},mark=otimes*\\%
lime!80!black,		densely dashed,	every mark/.append style={fill=lime},mark=otimes*\\%
red,                solid, every mark/.append style={fill=red!80!black},mark=otimes*\\%
cyan!60!black, 		densely dotted,	every mark/.append style={fill=cyan!80!black},mark=diamond*\\%
lime!80!black, 		densely dashed,	every mark/.append style={fill=lime},mark=diamond*\\%
red,           		solid, every mark/.append style={solid,fill=red!80!black},mark=diamond*\\%
teal,				every mark/.append style={fill=teal!80!black},mark=*\\%
orange,				every mark/.append style={fill=orange!80!black},mark=square*\\%
red!70!white,		mark=star\\%
yellow!60!black,		densely dashed, every mark/.append style={solid,fill=yellow!80!black},mark=square*\\%
black,				every mark/.append style={solid,fill=gray},mark=otimes*\\%
blue,				densely dashed,mark=star,every mark/.append style=solid\\%
red,					densely dashed,every mark/.append style={solid,fill=red!80!black},mark=diamond*\\%
}

\graphicspath{{figures/}}

\bgroup
\makeglossaries
\newacronym{ADMM}{ADMM}{Alternating Direction Method of Multipliers}
\newacronym{APC}{APC}{area power consumption}
\newacronym{ASE}{ASE}{area spectral efficiency}
\newacronym{BS}{BS}{base station}
\newacronym{CDN}{CDN}{content delivery network}
\newacronym{CN}{CN}{core network}
\newacronym{ICN}{ICN}{information-centric network}
\newacronym{CF}{CF}{collaborative filtering}
\newacronym{CRP}{CRP}{{C}hinese restaurant process}
\newacronym{CS}{CS}{central scheduler}
\newacronym{D2D}{D2D}{device-to-device}
\newacronym{EE}{EE}{energy efficiency}
\newacronym{ICIC}{ICIC}{inter-cell interference coordination}
\newacronym{LTE}{LTE}{long term evolution}
\newacronym{MIMO}{MIMO}{multiple-input multiple-output}
\newacronym{PPP}{PPP}{{P}oisson point process}
\newacronym{SBS}{SBS}{small base station}
\newacronym{SINR}{SINR}{signal-to-interference-plus-noise ratio}
\newacronym{SCN}{SCN}{small cell network}
\newacronym{SVD}{SVD}{singular value decomposition}
\newacronym{UT}{UT}{user terminal}
\newacronym{QoS}{QoS}{quality-of-service}
\newacronym{QoE}{QoE}{quality-of-experience}
\newacronym{RAN}{RAN}{radio access network}
\newacronym{PDF}{PDF}{probability distribution function}
\newacronym{PGFL}{PGFL}{probability generating functional}
\newacronym{HetNet}{HetNet}{heterogeneous network}

\newtheorem{proposition}{Proposition}

\begin{document}
\title{Caching at the Edge: \\ a Green Perspective for $5$G Networks}
\author{
		\IEEEauthorblockN{Bhanukiran Perabathini$^{\star, \dagger}$, Ejder Baştuğ$^{\diamond}$, Marios Kountouris$^{\star, \diamond}$, Mérouane Debbah$^{\diamond}$ and Alberto Conte$^{\dagger}$}
		\IEEEauthorblockA{
				\vspace{-0.25cm}
				\\
				\small
				$^{\diamond}$Large Networks and Systems Group (LANEAS), CentraleSupélec, 91192, Gif-sur-Yvette, France \\	
				$^{\star}$Department of Telecommunications, CentraleSupélec, 91192, Gif-sur-Yvette, France \\
				$^{\dagger}$Alcatel Lucent Bell Labs, 91620, Nozay, France \\
				\{ejder.bastug, marios.kountouris, merouane.debbah\}@supelec.fr,\\ \{bhanukiran.perabathini, alberto.conte\}@alcatel-lucent.com \\
				\vspace{-0.65cm}
		}
		\thanks{This research has been supported by the ERC Starting Grant 305123 MORE (Advanced Mathematical Tools for Complex Network Engineering), the projects 4GinVitro and BESTCOM.}
}
\IEEEoverridecommandlockouts
\maketitle

\begin{abstract}
Endowed with context-awareness and proactive capabilities, caching users' content locally at the edge of the network is able to cope with increasing data traffic demand in 5G wireless networks. In this work, we focus on the energy consumption aspects of cache-enabled wireless cellular networks, specifically in terms of area power consumption (APC) and energy efficiency (EE). We assume that both base stations (BSs) and mobile users are distributed according to homogeneous Poisson point processes (PPPs) and we introduce a detailed power model that takes into account caching. We study the conditions under which the area power consumption is minimized with respect to BS transmit power, while ensuring a certain quality of service (QoS) in terms of coverage probability. Furthermore, we provide the optimal BS transmit power that maximizes the area spectral efficiency per unit total power spent. The main takeaway of this paper is that caching seems to be an energy efficient solution.
\end{abstract}
\section{Introduction}
\label{sec:introduction}
Fueled by the ubiquity of new wireless devices and smart phones, as well as the proliferation of bandwidth-intensive applications, user demand for wireless data traffic has been growing tremendously, increasing the corresponding network load in an exponential manner. This is further exacerbated by rich media applications associated with video streaming and social networking. 
In parallel, both academia and industry are now in the urge of evolving traditional cellular networks towards the next-generation broadband mobile networks, coined as 5G networks, targeting to satisfy the mobile data tsunami while minimizing expenditures and energy consumption. Among these intensive efforts, caching users' content locally at the edge of the network is considered as one of the most disruptive paradigms in 5G networks \cite{Andrews2014Will}.

Interestingly yet not surprising, recent results have shown that distributed content caching can significantly offload different parts of the network, such as in \glspl{RAN} and \ac{CN}, by smartly prefetching and storing content closer to the end-users \cite{Bastug2014LivingOnTheEdge}. Indeed, traditional cellular networks, which have been designed for mobile devices with limited processing and storage capabilities, have started incorporating context-aware and proactive capabilities, fueled by recent advancements in processing power and storage. As a result, caching has recently taken the 5G literature by storm.

The idea of caching at the edge of the network, namely at the level of \glspl{BS} and \glspl{UT}, have been highlighted in various works, including edge caching \cite{Goebbels2010Disruption}, FemtoCaching \cite{Golrezaei2012Femtocaching}, and proactive caching \cite{Bastug2013Proactive}. Spatially random distributed cache-enabled \glspl{BS} are modeled in \cite{Bastug2014CacheEnabled}, and expressions for the outage probability and average delivery rate are derived therein. Another stochastic framework but in the context of cache-enabled \ac{D2D} communications is presented in \cite{Altieri2014Fundamental}, studying performance metrics that quantify the local and global fraction of served content requests. 
Although most prior work (see \cite{Bastug2015ThinkBefore} for a recent survey) deals with different aspects of caching, i.e. performance characterization and approximate algorithms, the energy consumption behavior of cache-enabled \glspl{BS} in densely deployed scenarios has not been investigated. This is precisely the focus of this paper.

Energy consumption aspects of cellular networks are usually investigated by placing \glspl{BS} on a regular hexagonal or grid topology and conducting intensive system-level simulations \cite{Jain2008System}. Therein, the aim is to find the optimal values of parameters, such as cell range and BS transmit power, while minimizing the total power consumed under certain \ac{QoS} constraints. Despite its attractiveness, it is rather cumbersome or even impossible to analytically evaluate key performance metrics in large-scale networks. A common simplification in modeling cellular networks is to place the \glspl{BS} on a two dimensional plane according to a homogeneous spatial \ac{PPP}, enabling to handle the problem analytically \cite{Andrews2011Tractable}. Several works studied the validity of \ac{PPP} modeling of \glspl{BS} compared to regular cellular models (e.g. \cite{Chen2012Small}) and additional insights can be obtained by analytically characterizing performance metrics, such as the spatial distribution of \ac{SINR}, coverage probability, and average rate \cite{Lee2013Stochastic, Mukherjee2014Analytical}.

In this work, we analyze energy consumption aspects of cache-enabled wireless network deployments using a spatial model based on stochastic geometry, which to the best of our knowledge has not been addressed in prior works. We consider the problem of \ac{APC} and \ac{EE} of cache-enabled \glspl{BS} in a scenario where \glspl{BS} and mobile users are distributed according to independent homogeneous \glspl{PPP}.
We provide conditions under which the area power consumption is minimized subject to a certain quality of service (QoS) in terms of coverage probability. Furthermore, we provide the optimal BS transmit power that maximizes the energy efficiency, defined as the ratio of area spectral efficiency over the total power consumed. The main result of this paper is that distributed content caching in wireless cellular networks turns out to be an energy efficient solution.

\section{System Model}
\label{sec:systemmodel}
We consider the downlink scenario of a single-tier cellular network in which the \glspl{BS} are distributed on the two-dimensional Euclidean plane $\mathbb{R}^{2}$ according to a homogeneous spatial \ac{PPP} with density $\lambda_{b}$ denoted by $\Phi_{b} = \{r_{i}\}_{i \in \mathbb{N}}$, where $r_{i} \in \mathbb{R}^{2}$ is the location of the $i$-th BS. The mobile \glspl{UT} (or users) are modeled to be distributed on the same plane according to an independent homogeneous \ac{PPP} with density $\lambda_{u}$ > $\lambda_{b}$, which is denoted by $\Phi_{u} = \{s_{j}\}_{j \in \mathbb{N}}$ with $s_{j} \in \mathbb{R}^{2}$ being the location of the $j$-th user. Without loss of generality, we focus on a \emph{typical} \ac{UT} placed at the origin of the coordinate system for calculating the key performance metrics of interest.

In this system model, we assume that the signal transmitted from a given \ac{BS} is subject to two propagation phenomena before reaching a user: (i) a distance dependent pathloss governed by the pathloss function $g(r) = b r^{-\alpha}$, where $b$ is the pathloss coefficient and $\alpha$ is the pathloss exponent, and (ii) Rayleigh fading with mean 1. Therefore, the signal strength from the $i$-th \ac{BS} as received by the typical user can be expressed as 
\begin{equation}
  \label{eq:sig-strength}
  p_{i}(r_{i}) = h_{i} P b r_{i}^{-\alpha},
\end{equation}
where the random variable $h_{i}$ denotes the power of Rayleigh fading and $P$ is the transmitted power. In addition, we assume that background noise is present in the system, with variance $\sigma^{2} = \beta \lambda_{b}$, with $\beta = B \frac{1}{\lambda_{u}} \frac{F k T}{b}$, where $B$ is the total available bandwidth, $F$ is the receiver noise figure, $k$ is Boltzmann constant, and $T$ is the ambient temperature. 

As alluded earlier, the typical user connected to the $i$-th \ac{BS} receives a signal of power $p_{i}(r_i)$. This in turn means that the sum of the received powers from the rest of the \glspl{BS} contributes to the interference to this signal. As a result, the received \ac{SINR} at the typical user is given by 
\begin{equation}
  \label{eq:sinr-intro}
  \text{SINR} = \frac{h_{i} g(r_{i}) P}{\sigma^{2} + I_{i}},
\end{equation}
where $I_{i} = \sum_{r_{j} \in \Phi_{b} \setminus r_{i}} p_{j}(r_{j})$ is the cumulative interference experienced from all the  \glspl{BS} except the  $i$-th \ac{BS}.

In a downlink scenario, although the typical user may technically be served from any \ac{BS}, a connection with a particular \ac{BS} has to be established according to a association policy that satisfies a certain performance metric, e.g. ensuring \ac{QoS}. In this work, we assume that a mobile user connects to the \ac{BS} that provides the maximum \ac{SINR}. This is formally expressed as
\[
	\underset{r_i \in \Phi_{b}}{\text{max}}{\text{SINR}(i) > \gamma},
\]
where $\gamma$ is the target \ac{SINR}. Given the above definition, the typical user is said to be \textit{covered} when there is at least one \ac{BS} that offers an $\text{SINR} > \gamma$. If not, we say that the typical user is not covered. We assume that $\gamma > 1$, which is needed to ensure that there is at maximum one \ac{BS} that provides the highest \ac{SINR} for a user at a given instant \cite{Dhillon2012Modeling}. 

\subsection{Coverage and Quality of Service}
Among the main objectives of wireless networks is to guarantee a certain \ac{QoS} for the users. The choice of a specific \ac{QoS} metric influences both the complexity and the operation of the network. In this paper, we consider the coverage probability as the \ac{QoS} that has to be assured. The coverage probability for strongest \ac{BS} association for a general pathloss function $g(r)$ is given as~\cite{Dhillon2012Modeling, Samarasinghe2013OptimalSinr}
\begin{eqnarray}
  \label{eq:Pcov}
P_{\text{cov}}(P, \lambda_{b}) &=& \mathbb{P}[\text{SINR} \geq \gamma] \nonumber \\
&=& \pi \lambda_{b} \int_{0}^{\infty} \exp(-q(P, \lambda_{b}, r)) \, dr,
\end{eqnarray}
where
\begin{equation}
  \label{eq:Pcovq}
q(P, \lambda_{b}, r) = \frac{\gamma \sigma^{2}}{P g(\sqrt{r})} +\lambda_{b} \int_{0}^{\infty} \frac{\pi \gamma g(\sqrt{r_{i}}) }{g(\sqrt{r}) + \gamma g(\sqrt{r_{i}})} \, dr_{i}.
\end{equation}

Denoting by $P_{\text{cov}}^{\text{NN}}$ the coverage probability in the case of no noise ($\sigma^2 \to 0$), the expression for the coverage probability in \eqref{eq:Pcov} may be approximated in a low noise regime as \cite{Perabathini2014Optimal}
\begin{equation}
  \label{eq:pcov-approx}
  P_{\text{cov}}(P, \lambda_{b}) = P_{\text{cov}}^{\text{NN}}\left(1 - \frac{A'}{\lambda_{b}^{\frac{\alpha}{2} - 1}}\frac{1}{P}\right),
\end{equation}
where $A' = \frac{\beta
 \Gamma(1 + \frac{\alpha}{2})}{b [C(\alpha)]^{\frac{\alpha}{2}}}$ and $C(\alpha) = \frac{2\pi^{2}}{\alpha} \mathrm{csc}\left({\frac{2\pi}{\alpha}}\right)$.

The specific \ac{QoS} constraint we consider here is that the coverage probability experienced by a typical user has to be always greater than the coverage probability achieved in the absence of noise, i.e.  
\begin{equation}
  \label{eq:QoS-constraint}
  P_{\text{cov}} \geq P_{\text{cov}}^{\text{NN}}.
\end{equation}
As a consequence, it can be shown (from \eqref{eq:pcov-approx} and \eqref{eq:QoS-constraint}) that the optimal \ac{BS} density $\lambda_{b}^{*}$ that guarantees the above \ac{QoS} constraint is
\begin{equation}
  \label{eq:lam-star}
  \lambda_{b}^{*} \geq \frac{A}{P^{\frac{2}{\alpha - 2}}},
\end{equation}
with $A = A'^{\frac{2}{\alpha - 2}}$. We use this expression in the optimization problem we formulate in Section~\ref{sec:OAPC}.

Once the typical user is covered with a certain \ac{SINR}, the traffic requests of this typical user have to be satisfied from its \ac{BS}, by bringing the content from its source on the Internet via the backhaul. In practice, even though the typical user is covered and benefits from high \ac{SINR}, it is obvious that any kind of bottleneck in the backhaul may result into long delays to the content, degrading the overall quality of experience (QoE). As part of dealing with this bottleneck, we assume that the \glspl{BS} are able to store the users' (popular) content in their caches, so that requests can be satisfied locally, without passing over the limited backhaul. This is detailed in the following section.
\subsection{Cache-enabled Base Stations}
Several studies have shown that multiple users actually access the same content very frequently. Take for instance some popular TV shows, the case of viral videos with over a billion viewings, news blogs, online streaming, etc. In this context, the network will be inundated with requests for the same content that might largely increase the latency or, eventually, congest the network itself. Otherwise stated, certain types of content (or information) are relatively more popular than others and are requested/accessed more often by the users \cite{Shafiq2011Characterizing}. Therefore, it is reasonable to assume that a user's choice distribution matches with the global content popularity distribution. As mentioned before, the logic behind having cache-enabled \glspl{BS} is to exploit this likelihood and store locally at the \glspl{BS} serving a typical user the content with highest demand (popular demand) so that both users and service providers get an incentive when a popular request is made. 

Let us assume that each \ac{BS} is equipped with a storage unit (hard disk) which caches popular content. Since the storage capacity cannot be infinite, we assume that at each \ac{BS} a set of content up to $f_{0}$ (the catalog) is stored on the hard disk. Rather than caching uniformly at random, a smarter approach will be to store the most popular content according to the given global content popularity statistics. We model the content popularity distribution at a typical user to be a right continuous and monotonically decreasing \ac{PDF}, denoted as \cite{Newman2005Power}
\begin{equation}
\label{eq:contentpdf}
f_{\mathrm{pop}}\left(f,\eta\right)
	=
	\begin{cases}
	\left(\eta - 1\right)f^{-\eta},
		& f \geq 1, \\
		0,			
		& f < 1,
	\end{cases}
\end{equation}
where $f$ is a point in the support of the corresponding content, and $\eta$ represents the steepness of the popularity distribution curve. We define the steepness factor to be the (average) number of users per \ac{BS}, that is $\eta = \frac{\lambda_{u}}{\lambda_{b}}$. The justification for the above model is that the higher the number of users attached to a \ac{BS}, the more accurately the trend is sampled, hence the more content is sorted towards the left of the distribution thereby making it steeper. Moreover, since $\lambda_{u} > \lambda_{b}$, we have that $\eta > 1$.

Now, given the fact that \glspl{BS} cache the catalog according to the content popularity distribution in (\ref{eq:contentpdf}), the probability that a content demanded by a connected user falls within the range $[0, f_{0}]$ is given by
\begin{eqnarray}
  \label{eq:prob-hit}
  \mathbb{P}_{\text{hit}} &=& \int_{0}^{f_{0}} f_{\text{pop}}(f, \eta) \, df \nonumber \\
  						  &=& \int_{0}^{f_{0}} (\eta - 1) f^{-\eta} \, df \nonumber \\
                       	  &=& 1 - f_{0}^{1 - \eta}.
\end{eqnarray}
It can be verified that $\mathbb{P}_{\text{hit}}$ converges to 1 when $f_{0} \rightarrow \infty$, namely when the  catalog stored in the \glspl{BS} goes to infinity. Consequently, the probability that a request is missing from the catalog can be expressed as $\mathbb{P}_{\text{miss}} = f_{0}^{1 - \eta}$.  An illustration of the system model is given in Fig. \ref{fig:scenario}, including snapshots of \glspl{PPP} and visualization of the content popularity distribution. In the following, we introduce a power model which takes into account the caching capabilities at \glspl{BS}, and will be used for investigating the energy aspects of cache-enabled \ac{BS} deployment.
\begin{figure*}[ht!]
	\centering
	\includegraphics[width=0.95\linewidth]{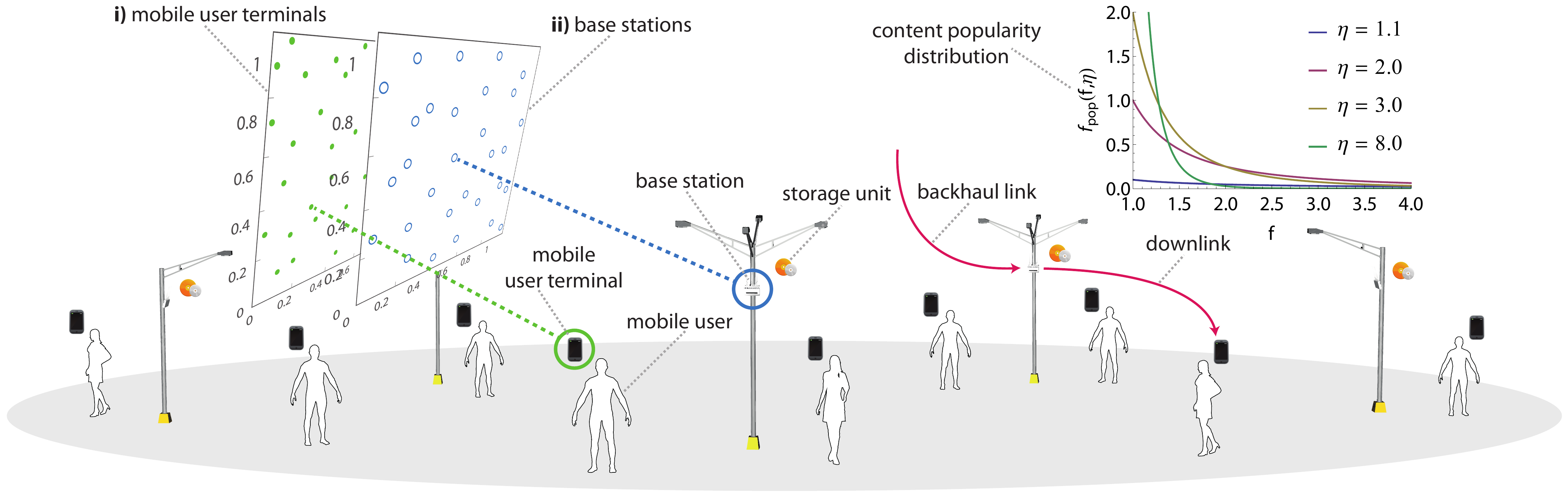}
	\caption{An illustration of the considered system model. The snapshots of PPPs for  i) mobile users and ii) base stations on unit area are given on the left side. The content popularity distribution for different values of $\eta$ is given on the top right side, showing that lower values of $\eta$ corresponds to a more uniform behavior.}
	\label{fig:scenario}
	\vspace{-0.2cm}
\end{figure*} 

\subsection{Power Consumption Model}
In the context of studying caching at the edge, there is a large scope for adopting an extensive power model that takes into account various detailed factors. In this paper, we deliberately restrict ourselves to a basic power model as a first attempt to relate energy efficiency to a cache enabled wireless network. We address a more complete model in an extended version of this paper. We consider two different power models depending on whether \glspl{BS} have caching capabilities or not.
\setcounter{subsubsection}{0}
\subsubsection{With caching}
\label{sec:power-with-cache}
The components of the total power consumed at an operating \ac{BS} is given as follows:
\begin{enumerate}
\item[i)] A constant transmit power, $P$.
\item[ii)] An operational charge at each \ac{BS}, $P_{\text{o}}$.
\item[iii)] Power needed to retrieve data from the local hard disk when a content from the catalog is requested, $P_{\text{hd}}$.
\item[iv)] Power needed to retrieve data from the backhaul when a content outside the catalog is requested, $P_{\text{bh}}$.
\end{enumerate}
We assume that $P_{\text{bh}} > P_{\text{hd}}$ motivated by the realistic constraint that it is more power consuming to utilize the backhaul connection than to retrieve stored information from the local caching/storage entity.

Therefore, the total power consumed at a given \ac{BS} is sum of all the components such as
\begin{eqnarray}
  \label{eq:P-total}
   P_{\text{tot}}^{\text{(c)}} &=& P + P_{o} + P_{\text{bh}} \times \mathbb{P}_{\text{miss}} +  P_{\text{hd}} \times \mathbb{P}_{\text{hit}} \nonumber \\
                         &=& P + P_{o} + P_{\text{hd}}  + (P_{\text{bh}} - P_{\text{hd}}) f_{0}^{1 - \eta} \nonumber \\
                         &=& P + P_{\text{s}}  + P_{\text{d}} f_{0}^{1 - \eta},
\end{eqnarray}
where $P_{\text{s}} = P_{o} + P_{\text{hd}}$ and $P_{\text{d}} =  P_{\text{bh}} - P_{\text{hd}}$.
%
\subsubsection{Without caching}
\label{sec:power-no-cache}
In the absence of caching, the \ac{BS} has to retrieve the requested content from the backhaul every service timeslot. This is equivalent to the case where $\mathbb{P}_{\text{hit}} = 0$ (or $\mathbb{P}_{\text{miss}} = 1$). The total power consumed at a given \ac{BS} is given as
\begin{eqnarray}
  \label{eq:1}
  P_{\text{tot}} &=& P + P_{o} + P_{\text{bh}} \nonumber \\
              &=& P + P_{s} + P_{\text{d}}.
\end{eqnarray}
\section{Area Power Consumption}
\label{sec:OAPC}
The power expenditure per unit area, also termed as \ac{APC}, is an important metric to characterize the deployment and operating costs of \glspl{BS}, also indicating the compatibility of the system with the legal regulations. In our system model, the \ac{APC} for cache-enabled \glspl{BS} is defined as
\begin{equation}
  \label{eq:APC-cache}
  \mathcal{P}^{\text{(c)}} = \lambda_{\text{b}} P_{\text{tot}}^{\text{(c)}}.
\end{equation}
In the same way, the \ac{APC} of \glspl{BS} with no caching capabilities is defined as $\mathcal{P} = \lambda_{\text{b}} P_{\text{tot}}$. In this section, we aim at minimizing separately the \ac{APC} for both cases (caching and no caching), while satisfying a certain \Ac{QoS}. This can be formally written as
\begin{equation}
  \label{eq:opt-prob}
  \begin{aligned}
        & \underset{P \in [0, \infty)}{\text{minimize}}&            & \mathcal{P}(P) \text{ or }  \mathcal{P}^{\text{(c)}}(P)\\
        & \text{subject to} &                              & P_{\text{cov}} (P, \lambda_{b}) \geq P_{\text{cov}}^{\text{NN}}.
  \end{aligned}
\end{equation}

Consider first the case where \glspl{BS} have caching capabilities. The following result can be obtained for the solution of the optimization problem in (\ref{eq:opt-prob}).
\begin{proposition}
   \label{prop:oapc}
   Suppose that \glspl{BS} have caching capabilities, thus $\mathcal{P}^{\text{(c)}}(P)$ is the objective (utility) function in (\ref{eq:opt-prob}). Then, for $\alpha = 4 + \epsilon$, $\epsilon > 0$, the optimal power allocation $P^{\star}$ that solves (\ref{eq:opt-prob}) is lower bounded as
 	\begin{equation}
		P^{*}  > \frac{2 P_{\text{s}}}{\epsilon}.
	\end{equation}
\end{proposition}
\begin{proof}
Using the expression for optimum $\lambda_{b}^{\star}$ from \eqref{eq:lam-star} and incorporating it in \eqref{eq:APC-cache}, we get the expression for \ac{APC} as
\begin{equation}
  \label{eq:op-APC-cache}
  \mathcal{P}^{(c)} = \frac{A}{P^{\frac{2}{\alpha - 2}}}(P + P_{\text{s}} + P_{\text{d}} f_{0}^{1-\frac{ \lambda_{u}}{A} P^{\frac{2}{\alpha - 2}}}).
\end{equation}

Let $\epsilon$ be a real number and write $\alpha = 4 + \epsilon$.

\begin{eqnarray}
  \label{eq:op-APC-cache2}
  \mathcal{P}^{(c)} &=& \frac{A}{P^{\frac{2}{2 + \epsilon}}}(P + P_{\text{s}} + P_{\text{d}} f_{0}^{1-\frac{ \lambda_{u}}{A} P^{\frac{2}{2 + \epsilon}}}) \nonumber \\
                   &=& A(P^{1 - \frac{2}{2 + \epsilon}} + P_{\text{s}} P^{-\frac{2}{2 + \epsilon}} + P_{\text{d}} P^{-\frac{2}{2 + \epsilon}} f_{0}^{1-\frac{ \lambda_{u}}{A} P^{\frac{2}{2 + \epsilon}}}). \nonumber \\
\end{eqnarray}

For $\epsilon \leq 0$, $\mathcal{P}^{(c)}$ is a monotonically decreasing function and no minimum point exists. However, for $\epsilon > 0$ (i.e. $\alpha > 4$), the first term in \eqref{eq:op-APC-cache2} dominates as $P \rightarrow \infty$, indicating that there exists a minimum where the derivative of the curve changes its sign from negative to positive. We set $\epsilon > 0$ for what follows.

Differentiating $\mathcal{P}^{(c)}$ with respect to $P$ we get
\begin{multline}
  \label{eq:APC-cache-derivative}
  \frac{d\,\mathcal{P}^{(c)}}{d\,P} = \frac{P^{-\frac{\epsilon +4}{\epsilon +2}}
   f_{0}^{-\frac{\lambda_{u} P^{\frac{2}{\epsilon
   +2}}}{A}}}{\epsilon +2} \times \\
A (P \epsilon -2 P_{\text{s}})
   f_{0}^{\frac{\lambda_{u} P^{\frac{2}{\epsilon
   +2}}}{A}}- 2 \lambda_{u} f_{0} P_{\text{d}} \log
   (f_{0}) P^{\frac{2}{\epsilon +2}} - 2 A f_{0} P_{\text{d}}. 
\end{multline}
Given the fact that $P$ is always positive, the derivative in \eqref{eq:APC-cache-derivative} remains negative
as $P$ increases from $0$ until it is sufficiently greater than $\frac{2 P_{\text{s}}}{\epsilon}$, after which the derivative can change its sign to positive. This indicates that there exists a minimum value for $\mathcal{P}^{(c)}$ when
\begin{equation}
  \label{eq:lower-lim}
   P^{*} > \frac{2 P_{\text{s}}}{\epsilon},
\end{equation}
which concludes the proof.
\end{proof}

For the case where \glspl{BS} have no caching capabilities, the following result is derived.
\begin{proposition}
Suppose that the \glspl{BS} have no caching capabilities, thus $\mathcal{P}(P)$ is the objective function in (\ref{eq:opt-prob}). Then, for $\alpha = 4 + \epsilon$, $\epsilon > 0$, the optimal power allocation $P^{\star}$ that solves (\ref{eq:opt-prob}) is given by 
	\begin{equation}
		 P^{*} = \frac{2(P_{\text{s}} + P_{\text{d}})}{\epsilon}.
	\end{equation}
\end{proposition}
\begin{proof}
In the case without caching, by similar treatment as in the proof of Proposition~\ref{prop:oapc}, we write the expression for \ac{APC} as
\begin{equation} 
  \label{eq:op-APC-nocache2}
  \mathcal{P} = \frac{A}{P^{\frac{2}{2 + \epsilon}}}(P + P_{\text{s}} + P_{\text{d}}).
\end{equation}
It can be noticed that $\mathcal{P}$ is a monotonically decreasing function and has no minimum except when $\epsilon > 0$ (or $\alpha > 4$).

Differentiating $\mathcal{P}$ with respect to $P$ we get
\begin{equation}
  \label{eq:APC-nocache-derivative}
   \frac{d\,\mathcal{P}}{d\,P} = \frac{A P^{-\frac{\epsilon +4}{\epsilon +2}} \left(-2 P_{\text{d}}-2 P_{\text{s}}+P
   \epsilon \right)}{\epsilon +2}.
\end{equation}

By equating the derivative to zero, the optimum power 
\begin{equation}
  \label{eq:P*-nocache}
  \frac{d\,\mathcal{P}}{d\,P} = 0 \Rightarrow P^{*} = \frac{2(P_{\text{s}} + P_{\text{d}})}{\epsilon}.
\end{equation}
It can easily be verified that the second derivative $\frac{d^{2}\,\mathcal{P}}{d\,P^{2}} > 0$ for ${P = P^{*}}$.
\end{proof}
\subsection{Remarks}
\label{sec:discussions}
For a given finite value of transmit power $P$, from \eqref{eq:op-APC-cache2} and \eqref{eq:op-APC-nocache2}, we observe that, in all sensible cases, 
\begin{equation}
  \label{eq:compare-APC}
  \mathcal{P}^{(c)} < \mathcal{P}.
\end{equation}
This indicates that \glspl{BS} with caching capabilities always outperform those without caching. Additionally, $\mathcal{P}^{(c)}$ can be made smaller by increasing the catalog size $f_{0}$ in \glspl{BS}. This is indeed intuitively correct, as shown in Fig.~\ref{fig:apcvsPvarf0} where some realistic power values from ~\cite{Imran2011Energy} are considered.

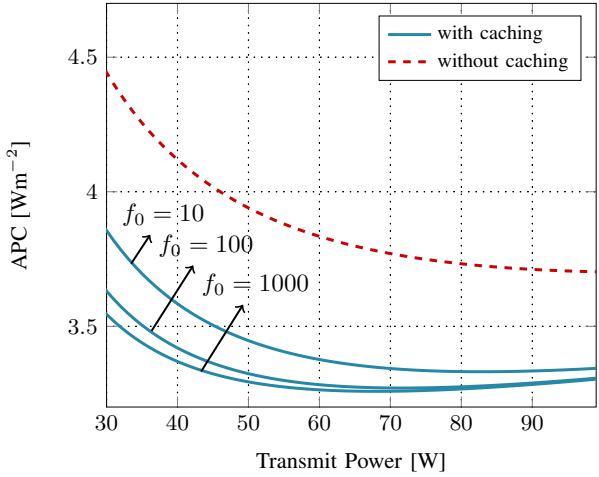
\begin{figure}[!ht]
\begin{tikzpicture}[scale=0.94]
	\begin{axis}[
 		grid = major,
 		legend cell align=left,
 		mark repeat={8},
 		ymin=3.2,ymax=4.7,
 		xmin=30,xmax=99,	
 		legend style ={legend pos=north east},
 		xlabel={Transmit Power [$\textup{W}$]},
 		ylabel={APC [$\textup{W}\textup{m}^{-2}$]}]

 		\addplot+[cyan!60!black, solid, very thick, mark=none, forget plot]
 				  table [col sep=comma] {\string"Pvarf0=10.0-withcaching-APC-vs-P-alpha-4.75.csv"};
 		\addlegendentry{with caching} 
		\node[anchor=east] (source) at (axis cs:33.900000, 3.700000){};
       	\node (destination) at (axis cs:38.00000, 3.913956){$f_0 = 10$};
       	\draw[thick,->](source)--(destination);
       	  
 		\addplot+[cyan!60!black, solid, very thick, mark=none, forget plot]
 				 table [col sep=comma] {\string"Pvarf0=100.0-withcaching-APC-vs-P-alpha-4.75.csv"};
		\node[anchor=east] (source) at (axis cs:36.800000, 3.45000){};
       	\node (destination) at (axis cs:44.000000, 3.8000){$f_0 = 100$};
       	\draw[thick,->](source)--(destination);

 		\addplot+[cyan!60!black, solid, very thick, mark=none] 
                                 table [col sep=comma] {\string"Pvarf0=1000.0-withcaching-APC-vs-P-alpha-4.75.csv"};
		\node[anchor=east] (source) at (axis cs:43.800000, 3.299364){};
       	\node (destination) at (axis cs:51.000000, 3.653956){$f_0 = 1000$};
       	\draw[thick,->](source)--(destination);
			 
		\addplot+[red!80!black, dashed, very thick, mark=none] 
                                 table [col sep=comma] {\string"Pvarf0-withoutcaching-APC-vs-P.csv"};
		\addlegendentry{without caching};
		 		  		
	\end{axis}
\end{tikzpicture}
\caption{\ac{APC} vs. Transmit power with and without caching for values: $P_{\text{s}} = 25 \textup{W}, P_{\text{d}} = 10 \textup{W}, \beta = 1$, and $\alpha = 4.75$ in \eqref{eq:op-APC-cache} and \eqref{eq:op-APC-nocache2}.}
\label{fig:apcvsPvarf0}
\vspace{-0.2cm}
\end{figure}

Fig. \ref{fig:apcvsP} illustrates the variation of \ac{APC} with respect to the transmit power in the cases with and without caching, for different values of pathloss exponent $\alpha$. It can be noticed that the \ac{APC} has a minimum for a certain power value only when $\alpha > 4$. However, \ac{APC} can be significantly reduced with caching in all cases, and the performance gap between caching and no caching cases is increased for $\alpha$ increasing. 

\begin{figure}[!ht]
\begin{tikzpicture}[scale=0.94]
	\begin{axis}[
 		grid = major,
 		legend cell align=left,
 		mark repeat={8},
 		ymin=0,ymax=40,
 		xmin=0,xmax=99,	
 		legend style ={legend pos=north east},
 		xlabel={Transmit Power [$\textup{W}$]},
 		ylabel={APC [$\textup{W}\textup{m}^{-2}$]}]

 		\addplot+[cyan!60!black, solid, very thick, mark=none] table [col sep=comma] {\string"P1-withcaching-APC-vs-P-alpha-4.00.csv"};
 		\addlegendentry{with caching} 
 		
 		\addplot+[cyan!60!black, solid, very thick, mark=none, forget plot]
 				  table [col sep=comma] {\string"P1-withcaching-APC-vs-P-alpha-5.00.csv"};
		 		  
 		\addplot+[cyan!60!black, solid, very thick, mark=none, forget plot]
 				 table [col sep=comma] {\string"P1-withcaching-APC-vs-P-alpha-6.00.csv"};
		 		   				 
		\addplot+[red!80!black, dashed, very thick, mark=none] table [col sep=comma] {\string"P1-withoutcaching-APC-vs-P-alpha-4.00.csv"};
		\addlegendentry{without caching};
		\node at (axis cs:50.000000, 0.933333) [anchor=south east] {\large $\alpha = 4$};
		 		  		
		\addplot+[red!80!black, dashed, very thick, mark=none, forget plot] table [col sep=comma] {\string"P1-withoutcaching-APC-vs-P-alpha-5.00.csv"};
		\node at (axis cs:50.000000, 5.920102) [anchor=south east] {\large $\alpha = 5$};
		
		\addplot+[red!80!black, dashed, very thick, mark=none, forget plot] table [col sep=comma] {\string"P1-withoutcaching-APC-vs-P-alpha-6.00.csv"};
		\node at (axis cs:50.000000, 15.014519) [anchor=south east] {\large $\alpha = 6$};
	\end{axis}
\end{tikzpicture}
\caption{\ac{APC} vs. Transmit power with and without caching for values: $P_{\text{s}} = 25 \textup{W}, P_{\text{d}} = 10 \textup{W}, f_{0} = 10, A = 2$ in \eqref{eq:op-APC-cache} and \eqref{eq:op-APC-nocache2}}
\label{fig:apcvsP}
\vspace{-0.2cm}
\end{figure}
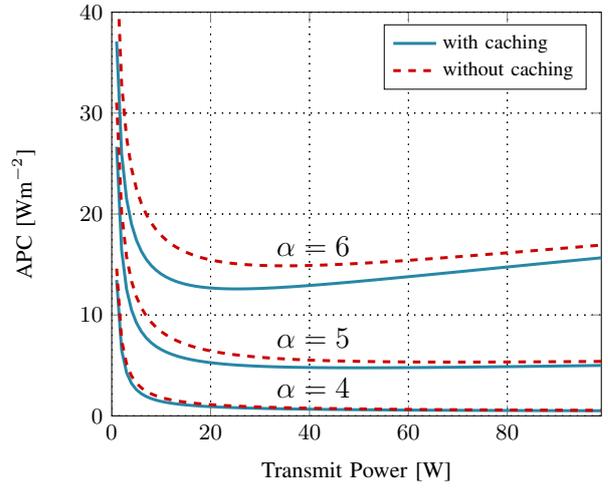


%
\section{Energy Efficiency}
\label{sec:OEE}
Another key performance metric that should be studied is the energy efficiency, which indicates the amount of utility (throughput) that is extracted out of a unit power invested for network operation. The standard \ac{QoS} factor chosen as the utility in literature is the \ac{ASE}. In \eqref{eq:EE-def}, we define the \ac{EE} as the ratio between \ac{ASE} and the total power spent at the \ac{BS}, i.e. 
\begin{equation}
  \label{eq:EE-def}
  \mathcal{E} = \frac{\lambda_{b} \log(1 + \gamma) P_{\text{cov}}(P, \lambda_{b})}{P_{\text{tot}}}.
\end{equation}
In order to define more precisely the EE metric for both cases, we use the expression for coverage probability established in \eqref{eq:pcov-approx}. With caching, the expression for \ac{EE} in \eqref{eq:EE-def} therefore becomes
\begin{equation}
  \label{eq:OEE-with-caching}
  \mathcal{E}^{(c)}(P, \lambda_{b}) = \frac{\lambda_{b} \log(1 + \gamma)  P_{\text{cov}}^{\text{NN}}(1 - \frac{A}{\lambda_{b}^{\frac{\alpha}{2} - 1}} \frac{1}{P})}{P + P_{\text{s}}  + P_{\text{d}} f_{0}^{1 - \frac{\lambda_{u}}{\lambda_{b}}}}.
\end{equation}
In the case without caching, \ac{EE} is given as
\begin{equation}
  \label{eq:OEE-without-caching}
  \mathcal{E}(P, \lambda_{b}) = \frac{\lambda_{b} \log(1 + \gamma) P_{\text{cov}}^{\text{NN}}(1 - \frac{A}{\lambda_{b}^{\frac{\alpha}{2} - 1}} \frac{1}{P})}{P + P_{\text{s}}  + P_{\text{d}}}.
\end{equation}
The following results are given for the maximization of \ac{EE} and the discussions are carried out afterwards.
\begin{proposition}
   Suppose that the \glspl{BS} have caching capabilities and let $P^{\star}_{1}$ denote the optimal power allocation that maximizes $\mathcal{E}^{(c)}(P, \lambda_{b})$. Then
   \begin{equation}
		P_{1}^{*} = 1 + \sqrt{1 + P_{\text{s}}  + P_{\text{d}} f_{0}^{1 - \frac{\lambda_{u}}{\lambda_{b}}}}.
   \end{equation}
\end{proposition}
\begin{proof}
We differentiate the expression \eqref{eq:OEE-with-caching} with respect to $P$ and solve for the optimum power $P_{1}^{*}$
\begin{eqnarray}
\label{eq:opt-EE1}
\frac{d\,\mathcal{E}^{(c)}}{d\,P} = 0 &\Rightarrow& 2 P - P^{2} + P_{\text{s}}  + P_{\text{d}} f_{0}^{1 - \frac{\lambda_{u}}{\lambda_{b}}} = 0 \nonumber \\
                               &\Rightarrow& P_{1}^{*} = 1 + \sqrt{1 + P_{\text{s}}  + P_{\text{d}} f_{0}^{1 - \frac{\lambda_{u}}{\lambda_{b}}}}.
\end{eqnarray}
It can be verified that the second derivative $\frac{d^{2}\,\mathcal{E}^{(c)}}{d\,P^{2}}$ is negative for the positive solution of $P_{1}^{*}$.
\end{proof}

\begin{proposition}
   Suppose that the \glspl{BS} have no caching capabilities and let $P^{\star}_{2}$ denote the optimal power allocation that maximizes $\mathcal{E}^{(c)}(P, \lambda_{b})$. Then
   \begin{equation}
		 P_{2}^{*} = 1 + \sqrt{1 + P_{\text{s}}  + P_{\text{d}}}.
   \end{equation}
\end{proposition}
\begin{proof}
Similar to \eqref{eq:opt-EE1}, we can show that the optimum power $P_{2}^{*} = 1 + \sqrt{1 + P_{\text{s}}  + P_{\text{d}}}$. It can be verified that the second derivative $\frac{d^{2}\,\mathcal{E}}{d\,P^{2}}$ is negative for the positive solution of $P_{2}^{*}$.
\end{proof}
\subsection{Remarks}
\label{sec:discussions-EE}
Based on the above results, we can make the following observations:
\begin{enumerate}
\item For given positive values of $\lambda_{b}$ and $P$, $\mathcal{E}^{(c)}(P, \lambda_{b})$ is always higher than $\mathcal{E}(P, \lambda_{b})$ making it apparent that implementing \glspl{BS} with caching capabilities is an energy-efficient solution.
\item For a fixed value of $\lambda_{b}$, $\mathcal{E}^{(c)}$ has a maximum at $P_{1}^{*} = 1 + \sqrt{1 + P_{\text{s}}  + P_{\text{d}} f_{0}^{1 - \frac{\lambda_{u}}{\lambda_{b}}}}$ and $\mathcal{E}(P, \lambda_{b})$ has a maximum at $P_{2}^{*} = 1 + \sqrt{1 + P_{\text{s}}  + P_{\text{d}}}$. Noting that $P_{1}^{*} < P_{2}^{*}$, we observe that the optimal \ac{EE} may be attained at a smaller value of transmit power in the case of cache-enabled \glspl{BS}.
\end{enumerate}

In Fig. \ref{fig:EEvsP} we plot the variation of \ac{EE} as a function of the transmit power $P$. It can be seen that \ac{EE} can be significantly increased (in the case of caching) by increasing the size of the catalog in \glspl{BS}, namely $f_{0}$.
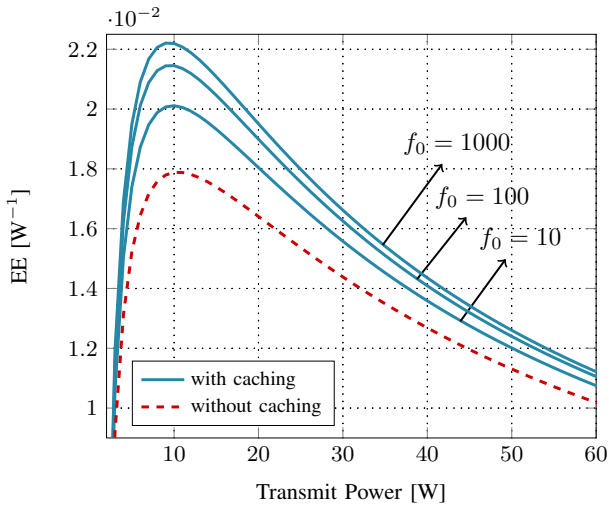
\begin{figure}[!ht]
\begin{tikzpicture}[scale=0.94]
	\begin{axis}[
 		grid = major,
 		legend cell align=left,
 		mark repeat={8},
 		ymin=0.009,ymax=0.0225,
 		xmin=2,xmax=60,	
 		legend style={at={(axis cs:5.0,0.0094)},anchor=south west},
 		xlabel={Transmit Power [$\textup{W}$]},
 		ylabel={EE [$\textup{W}^{-1}$]}]

 		\addplot+[cyan!60!black, solid, very thick, mark=none] table [col sep=comma] {\string"P2-withcaching-EE-vs-P-f0-10.csv"};
 		\addlegendentry{with caching};
		\node[anchor=north east] (source) at (axis cs:44.200000, 0.01290){};
       	\node (destination) at (axis cs:51.000000, 0.015647){$f_0 = 10$};
       	\draw[thick,->](source)--(destination);
		
 		\addplot+[cyan!60!black, solid, very thick, mark=none, forget plot] table [col sep=comma] {\string"P2-withcaching-EE-vs-P-f0-100.csv"};
		\node[anchor=north east] (source) at (axis cs:39.000000, 0.01431){};
       	\node (destination) at (axis cs:46.500000, 0.017049){$f_0 = 100$};
       	\draw[thick,->](source)--(destination);
       	
 		\addplot+[cyan!60!black, solid, very thick, mark=none, forget plot] table [col sep=comma] {\string"P2-withcaching-EE-vs-P-f0-1000.csv"}; 
		\node[anchor=north east] (source) at (axis cs:35.000000, 0.015453){};
       	\node (destination) at (axis cs:43.500000, 0.01883){$f_0 = 1000$};
       	\draw[thick,->](source)--(destination);
		
		\addplot+[red!80!black, dashed, very thick, mark=none] table [col sep=comma] {\string"P2-withoutcaching-EE-vs-P.csv"};
		\addlegendentry{without caching}
		
	\end{axis}
\end{tikzpicture}
\caption{\ac{EE} vs. transmit power with and without caching for values: $P_{\text{s}} = 25 \textup{W}, P_{\text{d}} = 10 \textup{W}, \alpha = 4.75, \beta = 1, P_{\text{cov}}^{\text{NN}} = 1, \lambda_{b} = 0.5, \lambda_{u} = 0.6, \text{ and } \gamma = 2$ in \eqref{eq:OEE-with-caching} and \eqref{eq:OEE-without-caching}.}
\label{fig:EEvsP}
\vspace{-0.2cm}
\end{figure}

\section{Conclusions}
\label{sec:confut}
In this work, we studied how incorporating caching capabilities at the BSs affects the energy consumption in wireless cellular networks. Adopting a detailed BS power model and modeling the BS locations according to a PPP, we derived expressions for the \ac{APC} and the \ac{EE}, which are further simplified in the low noise regime. A key observation of this work is that cache-enabled BSs can significantly decrease the \ac{APC} and improve the \ac{EE} as compared to traditional BSs. We also observed that the existence of an optimum power consumption point for the \ac{APC} depends on the pathloss exponent.

The energy aspects and implications of caching in wireless cellular networks, especially for 5G systems, are of practical and timely interest and clearly require further investigation. Future work may include heterogeneous network scenarios, including small cells, macro cells and WiFi access points deployment. Furthermore, storing the popular content requires accurate estimation of content popularity distribution, which cannot be easily performed in practice and may cost energy in terms of processing power. Therefore, rather than relying on this approach, randomized caching policies in a stochastic scenario \cite{Blaszczyszyn2014Geographic} can be considered as a means to provide crisp insights on the energy efficiency benefits of caching in dense wireless networks.
\bibliographystyle{IEEEtran}
\bibliography{references}

\end{document}